\documentclass[paper,a4paper]{IEEEtran}

\usepackage{amsthm}
\usepackage{amsmath}
\usepackage{amsfonts,amssymb,amsmath}
\usepackage{graphicx}
\usepackage{epsfig}
\usepackage{caption}
\usepackage[applemac]{inputenc}
\usepackage{latexsym}
\usepackage{dsfont}
\usepackage{tikz}
\usetikzlibrary{arrows,shapes,backgrounds,patterns,decorations.pathreplacing}

\newtheorem{definition}{Definition}
\newtheorem{corollary}{Corollary}
\newtheorem{proposition}{Proposition}
\newtheorem{theorem}{Theorem}
\newtheorem{lemma}{Lemma}

\newtheorem{remark}{Remark}

\DeclareMathOperator{\DMC}{DMC}

\DeclareMathOperator{\supp}{supp}

\DeclareMathOperator{\MP}{MP}

\begin{document}

\sloppy

\title{On the Convergence of the Polarization Process in the Noisiness/Weak-$\ast$ Topology}

\author{
  \IEEEauthorblockN{Rajai Nasser\\}
  \IEEEauthorblockA{
Email: rajai.nasser@alumni.epfl.ch} 
}



\maketitle

\begin{abstract}
Let $W$ be a channel where the input alphabet is endowed with an Abelian group operation, and let $(W_n)_{n\geq 0}$ be Ar{\i}kan's channel-valued polarization process that is obtained from $W$ using this operation. We prove that the process $(W_n)_{n\geq 0}$ converges almost surely to deterministic homomorphism channels in the noisiness/weak-$\ast$ topology. This provides a simple proof of multilevel polarization for a large family of channels, containing among others, discrete memoryless channels (DMC), and channels with continuous output alphabets. This also shows that any continuous channel functional converges almost surely (even if the functional does not induce a submartingale or a supermartingale).
\end{abstract}

\section{Introduction}
Polar codes are a family of capacity-achieving codes which were first introduced by Ar{\i}kan for binary input channels \cite{Arikan}. The construction of polar codes relies on a phenomenon called \emph{polarization}: A collection of independent copies of a channel is converted into a collection of synthetic channels that are extreme, i.e., almost useless or almost perfect.

The construction of polar codes was later generalized for channels with arbitrary (but finite) input alphabet \cite{SasogluTelAri,ParkBarg,SahebiPradhan,RajTelA,RajErgI,RajErgII}. Note that for channels where the input alphabet size is not prime, polarization is not necessarily a two-level polarization (to useless and perfect channels): We may have multilevel polarization where the polarized channels can be neither useless nor perfect.

In this paper, we are interested in the general multilevel polarization phenomenon which happens when we apply an Ar{\i}kan-style transformation that is based on an Abelian group operation. It was shown in \cite{SahebiPradhan} that as the number of polarization steps becomes large, the behavior of synthetic channels resembles that of deterministic homomorphism channels projecting their input onto a quotient group. This resemblance was formulated in \cite{SahebiPradhan} using Bhattacharyya parameters. We may say (informally), that as the number of polarization steps goes to infinity, the synthetic channels ``converge" to deterministic homomorphism channels. One reason why this statement is informal is because the synthetic channels do not have the same output alphabet, so in order to make the statement formal, we must define a space in which we can topologically compare channels with different output alphabets.

In \cite{RajTop}, we defined the space of all channels with fixed input alphabet and arbitrary but finite output alphabet. This space was first quotiented by an equivalence relation, and then several topological structures were defined. In this paper, we show that Ar{\i}kan's polarization process does converge in the noisiness/weak-$\ast$ topology to deterministic homomorphism channels. The proof uses the Blackwell measure of channels\footnote{Blackwell measures was used in \cite{RaginskyBlackwellPolar} and \cite{GoelaRaginsky} for analyzing the polarization of binary-input channels.}, and hence can be generalized to all channels whose equivalence class can be determined by the Blackwell measure. This family of channels contains, among others, all discrete memoryless channels and all channels with continuous output alphabets. Another advantage of our proof is that it implies the convergence of all channel functionals that are continuous in the noisiness/weak-$\ast$ topology. Therefore, we have convergence of those functionals even if they do not induce a submartingale or a supermartingale process.

In Section \ref{secPreliminaries}, we introduce the preliminaries of this paper. In Section \ref{secPolarPhen}, we recall the multilevel polarization phenomenon. In Section \ref{secConvergence}, we show the convergence of the polarization process in the noisiness/weak-$\ast$ topology. For simplicity, we only discuss discrete memoryless channels, but the proof is valid for any channel whose equivalence class is determined by the Blackwell measure.

\section{Preliminaries}

\label{secPreliminaries}

\subsection{Meta-Probability Measures}

Let $\mathcal{X}$ be a finite set. The set of probability distributions on $\mathcal{X}$ is denoted as $\Delta_\mathcal{X}$. We associate $\Delta_{\mathcal{X}}$ with its Borel $\sigma$-algebra. A \emph{meta-probability measure} on $\mathcal{X}$ is a probability measure on the Borel sets of $\Delta_{\mathcal{X}}$. It is called a meta-probability measure because it is a probability measure on the set of probability distributions on $\mathcal{X}$. We denote the set of meta-probability measures on $\mathcal{X}$ as $\mathcal{MP}(\mathcal{X})$.

A meta-probability measure $\MP$ on $\mathcal{X}$ is said to be \emph{balanced} if
$$\int_{\Delta_{\mathcal{X}}}p\cdot d\MP(p)=\pi_\mathcal{X},$$
where $\pi_\mathcal{X}$ is the uniform probability distribution on $\mathcal{X}$. The set of balanced meta-probability measures on $\mathcal{X}$ is denoted as $\mathcal{MP}_b(\mathcal{X})$. The set of balanced and finitely supported meta-probability measures on $\mathcal{X}$ is denoted as $\mathcal{MP}_{bf}(\mathcal{X})$.

\subsection{$\DMC$ Spaces}

Let $\mathcal{X}$ and $\mathcal{Y}$ be two finite sets. The set of discrete memoryless channels (DMC) with input alphabet $\mathcal{X}$ and output alphabet $\mathcal{Y}$ is denoted as $\DMC_{\mathcal{X},\mathcal{Y}}$. The set of channels with input alphabet $\mathcal{X}$ is defined as
$$\DMC_{\mathcal{X},\ast}=\coprod_{n\geq 1}\DMC_{\mathcal{X},[n]},$$
where $[n]=\{1,\ldots,n\}$ and $\coprod$ is the disjoint union symbol. The $\ast$ symbol in $\DMC_{\mathcal{X},\ast}$ means that the output alphabet is arbitrary but finite.

Let $\mathcal{X}$, $\mathcal{Y}$ and $\mathcal{Z}$ be three finite sets. Let $W\in\DMC_{\mathcal{X},\mathcal{Y}}$ and $V\in\DMC_{\mathcal{Y},\mathcal{Z}}$. The composition of $V$ and $W$ is the channel $V\circ W\in\DMC_{\mathcal{X},\mathcal{Z}}$ defined as:
$$(V\circ W)(z|x)=\sum_{y\in\mathcal{Y}}V(z|y)W(y|x).$$

A channel $W\in\DMC_{\mathcal{X},\mathcal{Y}}$ is said to be \emph{degraded} from another channel $W'\in\DMC_{\mathcal{X},\mathcal{Y}'}$ if there exists a channel $V'\in\DMC_{\mathcal{Y}',\mathcal{Y}}$ such that $W=V'\circ W'$. Two channels are said to be \emph{equivalent} if each one is degraded from the other. It is well known that if two channels are equivalent then every code has the same probability of error (under ML decoding) for both channels. This is why it makes sense from an information-theoretic point of view to identify equivalent channels and consider them as one object in the ``space of equivalent channels". The quotient of $\DMC_{\mathcal{X},\ast}$ by the equivalence relation is denoted as $\DMC_{\mathcal{X},\ast}^{(o)}$. The equivalence class of a channel $W\in\DMC_{\mathcal{X},\ast}$ is denoted as $\hat{W}$.

\subsection{A Necessary and Sufficient Condition for Degradedness}

Let $\mathcal{U},\mathcal{X},\mathcal{Y}$ be three finite sets and let $W\in\DMC_{\mathcal{X},\mathcal{Y}}$. For every $p\in\Delta_{\mathcal{U}\times\mathcal{X}}$, define
$$P_c(p,W)=\sup_{D\in\DMC_{\mathcal{Y},\mathcal{U}}}\sum_{\substack{u\in\mathcal{U},\\x\in\mathcal{X},\\y\in\mathcal{Y}}}p(u,x)W(y|x)D(u|y).$$

$P_c(p,W)$ can be interpreted as follows: Let $(U,X)$ be a pair of random variables in $\mathcal{U}\times\mathcal{X}$. Send $X$ through the channel $W$ and let $Y$ be the output. $P_c(p,W)$ can be seen as the optimal probability of correctly guessing $U$ from $Y$ among all random decoders $D\in\DMC_{\mathcal{Y},\mathcal{U}}$.

Now let $\mathcal{X},\mathcal{Y},\mathcal{Y}'$ be three finite sets and let $W\in\DMC_{\mathcal{X},\mathcal{Y}}$ and $W'\in\DMC_{\mathcal{X},\mathcal{Y}'}$. Buscemi proved in \cite{Buscemi} that $W$ is degraded from $W'$ if and only if $P_c(p,W)\leq P_c(p,W')$ for every $p\in\Delta_{\mathcal{U}\times\mathcal{X}}$ and every finite set $\mathcal{U}$. This means that $W$ and $W'$ are equivalent if and only if $P_c(p,W)= P_c(p,W')$ for every $p\in\Delta_{\mathcal{U}\times\mathcal{X}}$ and every finite set $\mathcal{U}$. Therefore, if $p\in\Delta_{\mathcal{U}\times\mathcal{X}}$ and $\hat{W}\in\DMC_{\mathcal{X},\ast}^{(o)}$, we can define $P_c(p,\hat{W})=P_c(p,W)$ for any $W\in\hat{W}$.

\subsection{Blackwell Measures}

\label{subsecBlackwell}

Let $W\in\DMC_{\mathcal{X},\mathcal{Y}}$. Let $(X,Y)$ be a pair of random variables in $\mathcal{X}\times\mathcal{Y}$, which is distributed as $$\mathbb{P}_{X,Y}(x,y)=\frac{1}{|\mathcal{X}|}W(y|x).$$ In other words, $X$ is uniformly distributed in $\mathcal{X}$ and $Y$ is the output of the channel $W$ when $X$ is the input. For every $y\in\mathcal{Y}$ satisfying $\mathbb{P}_Y(y)>0$, let $W_y^{-1}\in\Delta_{\mathcal{X}}$ be the posterior probability distribution of the input assuming $y$ was received. More precisely,
$$W_y^{-1}(x)=\mathbb{P}_{X|Y}(x|y)=\frac{W(y|x)}{\sum_{x'\in\mathcal{X}}W(y|x')}.$$
The \emph{Blackwell measure} of $W$ is the meta-probability measure $\MP_W\in\mathcal{MP}(\mathcal{X})$, which describes the random variable $W_Y^{-1}$. It is easy to see that
$$\MP_W=\sum_{\substack{y\in\mathcal{Y}:\\\mathbb{P}_Y(y)>0}}\mathbb{P}_Y(y)\delta_{W_y^{-1}}\in\mathcal{MP}(\mathcal{X}).$$

The following proposition, which is easy to prove, characterizes the Blackwell measures of DMCs:

\begin{proposition}
\cite{torgersen}\label{propBijection} A meta-probability measure $\MP\in\mathcal{MP}(\mathcal{X})$ is the Blackwell measure of a DMC with input alphabet $\mathcal{X}$ if and only if $\MP\in\mathcal{MP}_{bf}(\mathcal{X})$.
\end{proposition}

The following proposition shows that the Blackwell measure characterizes the equivalence class of a channel:

\begin{proposition}
\cite{torgersen} \label{propBlackwellEquivalence} Two channels $W\in\DMC_{\mathcal{X},\mathcal{Y}}$ and $W'\in\DMC_{\mathcal{X},\mathcal{Y}'}$ are equivalent if and only if $\MP_W=\MP_{W'}$.
\end{proposition}

For every $\hat{W}\in\DMC_{\mathcal{X},\ast}^{(o)}$, define $\MP_{\hat{W}}=\MP_W$ for any $W\in\hat{W}$. Proposition \ref{propBlackwellEquivalence} shows that $\MP_{\hat{W}}$ is well defined.

\subsection{The Noisiness/Weak-$\ast$ Topology}

In \cite{RajTop}, we defined the \emph{noisiness metric} $d_{\mathcal{X},\ast}^{(o)}$ on $\DMC_{\mathcal{X},\ast}^{(o)}$ as follows:
$$d_{\mathcal{X},\ast}^{(o)}(\hat{W},\hat{W}')=\sup_{\substack{m\geq 1,\\p\in\Delta_{[m]\times\mathcal{X}}}}|P_c(p,\hat{W})-P_c(p,\hat{W}')|,$$
where $[m]=\{1,\ldots,m\}$.

$d_{\mathcal{X},\ast}^{(o)}(\hat{W},\hat{W}')$ is called the noisiness metric because it compares the ``noisiness" of $\hat{W}$ with that of $\hat{W}'$: If $P_c(p,\hat{W})$ is close to $P_c(p,\hat{W}')$ for every random encoder $p$, then $\hat{W}$ and $\hat{W}'$ have close ``noisiness levels".

The topology on $\DMC_{\mathcal{X},\ast}^{(o)}$ which is induced by the metric $d_{\mathcal{X},\ast}^{(o)}$ is denoted as $\mathcal{T}_{\mathcal{X},\ast}^{(o)}$.

Another way to ``topologize" the space $\DMC_{\mathcal{X},\ast}^{(o)}$ is through Blackwell measures: Proposition \ref{propBijection} implies that the mapping $\hat{W}\rightarrow\MP_{\hat{W}}$ is a bijection from $\DMC_{\mathcal{X},\ast}^{(o)}$ to $\mathcal{MP}_{bf}(\mathcal{X})$. We call this mapping the \emph{canonical bijection} from $\DMC_{\mathcal{X},\ast}^{(o)}$ to $\mathcal{MP}_{bf}(\mathcal{X})$. By choosing a topology on $\mathcal{MP}_{bf}(\mathcal{X})$, we can construct a topology on $\DMC_{\mathcal{X},\ast}^{(o)}$ through the canonical bijection. We showed in \cite{RajTop} that the weak-$\ast$ topology is exactly the same as $\mathcal{T}_{\mathcal{X},\ast}^{(o)}$. This is why we call $\mathcal{T}_{\mathcal{X},\ast}^{(o)}$ the \emph{noisiness/weak-$\ast$ topology}.

\begin{remark}
\label{remIdent}
Since we identify $\DMC_{\mathcal{X},\ast}^{(o)}$ and $\mathcal{MP}_{bf}(\mathcal{X})$ through the canonical bijection, we can use $d_{\mathcal{X},\ast}^{(o)}$ to define a metric on $\mathcal{MP}_{bf}(\mathcal{X})$. Furthermore, since $\mathcal{MP}_{bf}(\mathcal{X})$ is dense in $\mathcal{MP}_{b}(\mathcal{X})$ (see e.g., \cite{RajTop}), we can extend the definition of $d_{\mathcal{X},\ast}^{(o)}$ to $\mathcal{MP}_{b}(\mathcal{X})$ by continuity. Similarly, we can extend the definition of any channel parameter or operation which is continuous in the noisiness/weak-$\ast$ topology (such as the symmetric capacity, Ar{\i}kan's polar transformations, etc. \cite{RajCont}) to $\mathcal{MP}_{b}(\mathcal{X})$.
\end{remark}

\section{The Polarization Phenomenon}

\label{secPolarPhen}

\subsection{Useful Notations}

Throughout this paper, $(G,+)$ denotes a finite Abelian group.

If $W$ is a channel, we denote the \emph{symmetric capacity}\footnote{The symmetric capacity of a channel is the mutual information between a uniformly distributed input and the output.} of $W$ as $I(W)$.

For every subgroup $H$ of $G$, define the channel $D_H\in\DMC_{G,G/H}$ as
$$D_H(A|x)=\begin{cases}1\quad&\text{if}\;x\in A,\\0\quad&\text{otherwise.}\end{cases}$$
In other words, $D_H$ is the deterministic channel where the output is the coset to which the input belongs. It is easy to see that $I(D_H)=\log|G/H|$.

We denote the set $\{D_H:\;H\;\text{is a subgroup of}\;G\}$ as $\mathcal{DH}_G$.

Now let $\mathcal{Y}$ be a finite set and let $W\in\DMC_{G,\mathcal{Y}}$. For every subgroup $H$ of $G$, define the channel $W[H]\in\DMC_{G/H,\mathcal{Y}}$ as
$$W[H](y|A)=\frac{1}{|A|}\sum_{x\in A}W(y|x)=\frac{1}{|H|}\sum_{x\in A}W(y|x).$$
\begin{remark}
Let $X$ be a random variable uniformly distributed in $G$ and let $Y$ be the output of the channel $W$. It is easy to see that $I(W[H])=I(X\bmod H,Y)$.
\end{remark}

Let $\delta>0$. We say that a channel $W\in\DMC_{G,\mathcal{Y}}$ is \emph{$\delta$-determined} by a subgroup $H$ of $(G,+)$ if $$\big|I(W)-\log|G/H|\big|<\delta\text{ and }\big|I(W[H])-\log|G/H|\big|<\delta.$$ We say that $W$ is \emph{$\delta$-determined} if there exists at least one subgroup $H$ which $\delta$-determines $W$. It is easy to see that if $\delta$ is small enough, there exists at most one subgroup that $\delta$-determines $W$.

Intuitively, if $\delta$ is small and $W$ is $\delta$-determined by a subgroup $H$, then the channel $W$ is ``almost-equivalent" to $D_H$: Let $X$ be a random variable that is uniformly distributed in $G$ and let $Y$ be the output of $W$ when $X$ is the input. We have:
\begin{itemize}
\item The inequality $\big|I(X\bmod H;Y)-\log|G/H|\big|<\delta$ means that $X\bmod H$ can be determined from $Y$ with high probability.
\item the inequality $\big|I(X;Y)-\log|G/H|\big|<\delta$ means that there is almost no other information about $X$ which can be determined from $Y$.
\end{itemize}
Due to the above two observations, we can (informally) say that if $W$ is $\delta$-determined by $H$, then $W$ is ``almost equivalent" to $D_H$.

\subsection{The Polarization Process}

Let $W\in\DMC_{G,\mathcal{Y}}$ be a channel with input alphabet $G$. Define the channels $W^-\in\DMC_{G,\mathcal{Y}^2}$ and $W^+\in\DMC_{G,\mathcal{Y}^2\times G}$ as follows:
$$W^-(y_1,y_2|u_1)=\frac{1}{|G|}\sum_{u_2\in G}W(y_1|u_1+u_2)W(y_2|u_2),$$
and
$$W^+(y_1,y_2,u_1|u_2)=\frac{1}{|G|}W(y_1|u_1+u_2)W(y_2|u_2).$$

For every $n\geq 1$ and every $s=(s_1,\ldots,s_n)\in\{-,+\}^n$, define $W^s=(\cdots((W^{s_1})^{s_2})\cdots)^{s_n}$.

\begin{remark}
It can be shown that if $W$ and $V$ are equivalent, then $W^-$ (resp. $W^+$) and $V^-$ (resp. $V^+$) are equivalent. This allows us to write $\hat{W}^-$ and $\hat{W}^+$ to denote $\widehat{W^-}$ and $\widehat{W^+}$, respectively.
\end{remark}

It was shown in \cite{ParkBarg}, \cite{SahebiPradhan} and \cite{RajTelA} that as $n$ becomes large, the behavior of almost all the synthetic channels $(W^s)_{s\in\{-,+\}^n}$ approaches the behavior of deterministic homomorphism channels projecting their input onto quotient groups.

One way to formalize the above statement was given in \cite{RajTelA} as follows: For every $\delta>0$, we have
\begin{equation}
\label{eqPolarization}
\begin{aligned}
\lim_{n\to\infty}\frac{1}{2^n}\big|\big\{s\in&\{-,+\}^n:\;W^s\;\text{is}\;\delta\text{-determined}\big\}\big|=1.
\end{aligned}
\end{equation}

\begin{definition}
\label{defPolProcess}
Let $(B_n)_{n\geq 1}$ be a sequence of independent and uniformly distributed Bernoulli random variables in $\{-,+\}$. Define the channel-valued random process $(W_n)_{n\geq 0}$ as follows:
\begin{itemize}
\item $W_0=W$.
\item $W_n=W_{n-1}^{B_n}=W^{(B_1,\ldots,B_n)}$ if $n\geq 1$.
\end{itemize}
\end{definition}

Equation \eqref{eqPolarization} can be rewritten as:
\begin{equation}
\label{eqPolarProb}
\lim_{n\to\infty}\mathbb{P}[\{W_n\;\text{is}\;\delta\text{-determined}\}]=1.
\end{equation}

One informal way to interpret Equation \eqref{eqPolarProb} is to say that ``the process $(W_n)_{n\geq 0}$ converges to channels in $\mathcal{DH}_G$". This statement will be made formal in the following section.

\section{Convergence of the Polarization Process}

\label{secConvergence}

Throughout this section, we identify a channel $W\in\DMC_{G,\ast}^{(o)}$ with its Blackwell measure $\MP_W\in\mathcal{MP}_{bf}(G)$. We also extend the definition of the $+$ and $-$ operations to all balanced measures in $\mathcal{MP}_b(G)$ (as discussed in Remark \ref{remIdent}).

\begin{lemma}
\label{lemConvergence}
Let $(W_n)_{n\geq 0}$ be the channel-valued process defined in Definition \ref{defPolProcess}. Almost surely, the sequence $\big(|I(W_n^-)-I(W_n)|\big)_{n\geq 0}$ converges to zero.
\end{lemma}
\begin{proof}
It is well known that $I(W^-)+I(W^+)=2I(W)$ for every channel with input alphabet $G$. Hence, we have
\begin{align*}
\mathbb{E}(I(W_{n+1})|W_n)=\frac{1}{2}I(W_n^-)+\frac{1}{2}I(W_n^+)=I(W_n).
\end{align*}
This shows that the process $(I(W_n))_{n\geq 0}$ is a martingale, hence it converges almost surely. This means that the process $\big(|I(W_{n+1})-I(W_n)|\big)_{n\geq 0}$ almost surely converges to zero. On the other hand, we have
\begin{align*}
|I(W_{n+1})-I(W_n)|&=\begin{cases}|I(W_n^-)-I(W_n)|\quad\text{if}\;B_{n+1}=-,\\|I(W_n^+)-I(W_n)|\quad\text{if}\;B_{n+1}=+,\end{cases}\\
&\stackrel{(a)}{=}|I(W_n^-)-I(W_n)|,\\
\end{align*}
where $(a)$ follows from the fact that $I(W_n^-)+I(W_n^+)=2I(W_n)$, which means that $|I(W_n^-)-I(W_n)|=|I(W_n^+)-I(W_n)|$.
\end{proof}

Now define the set
$$\mathcal{POL}_G=\{\MP\in\mathcal{MP}_b(G):\;I(\MP^-)=I(\MP)\}.$$

The next lemma shows that the Blackwell measures of channels with a small $|I(W^-)-I(W)|$ are close (in the noisiness metric sense) to measures in $\mathcal{POL}_G$.

\begin{lemma}
\label{lemAlmostExtreme}
For every $\epsilon>0$, there exists $\delta>0$ such that for every $\MP\in\mathcal{MP}_b(G)$, if $|I(\MP^-)-I(\MP)|<\delta$, then $$d_{G,\ast}^{(o)}(\MP,\mathcal{POL}_G)<\epsilon.$$
\end{lemma}
\begin{proof}
Define the function $f:\mathcal{MP}_b(G)\rightarrow \mathbb{R}^+$ as
$$f(\MP)=|I(\MP^-)-I(\MP)|.$$
Since the symmetric capacity and the $-$ transformation are continuous in the noisiness/weak-$\ast$ topology (see \cite{RajCont}), the function $f$ is also continuous in the same topology.

Since $\mathcal{POL}_G=f^{-1}(\{0\})$ and since $f$ is continuous, the set $\mathcal{POL}_G$ is closed.

Now for every $\epsilon>0$, define the set
$$\mathcal{POL}_{G,\epsilon}=\{\MP\in\mathcal{MP}_b(G):\;d_{G,\ast}^{(o)}(\MP,\mathcal{POL}_G)<\epsilon\},$$
and let 
\begin{align*}
\delta&=\inf(f(\mathcal{POL}_{G,\epsilon}^c))\\
&=\resizebox{0.44\textwidth}{!}{$\inf \big\{f(\MP):\;\MP\in\mathcal{MP}_b(G)\text{ and }d_{G,\ast}^{(o)}(\MP,\mathcal{POL}_G)\geq\epsilon\big\}$}.
\end{align*}

Since the set $\mathcal{POL}_G$ is closed, we can see that the set $\mathcal{POL}_{G,\epsilon}^c$ is closed as well. Furthermore, since the space $\mathcal{MP}_b(G)$ is compact (see e.g., \cite{RajTop}), the set $\mathcal{POL}_{G,\epsilon}^c$ is compact as well. Therefore, the set $f(\mathcal{POL}_{G,\epsilon}^c)$ is compact in $\mathbb{R}^+$, which means that its infimum is achieved, i.e., there exists $\MP_\epsilon\in\mathcal{POL}_{G,\epsilon}^c$ such that $\delta=f(\MP_\epsilon)$. But $\mathcal{POL}_{G,\epsilon}^c\cap \mathcal{POL}_G=\emptyset$ and $\mathcal{POL}_G=f^{-1}(\{0\})$, so we must have $\delta=f(\MP_\epsilon)>0$.

From the definition of $\delta$, we have:
$$d_{G,\ast}^{(o)}(\MP,\mathcal{POL}_G)\geq\epsilon\;\Rightarrow\;f(\MP)\geq\delta.$$
Hence, by contraposition, we have
$$f(\MP)<\delta\;\Rightarrow\;d_{G,\ast}^{(o)}(\MP,\mathcal{POL}_G)<\epsilon.$$
\end{proof}

In the rest of this section, we analyze the balanced meta-probability measures that are in $\mathcal{POL}_G$ (i.e., those that satisfy $I(\MP^-)=I(\MP)$).

For every $p\in\Delta_G$, we denote the entropy of $p$ as $H(p)$.

For every $p,q\in\Delta_G$, define $p\circledast q\in\Delta_G$ as follows:
$$(p\circledast q)(u_1)=\sum_{u_2\in G}p(u_1+u_2)q(u_2).$$

\begin{lemma}
\label{lemMainEq}
For every $\MP\in\mathcal{MP}_b(G)$, we have
\begin{align*}
|I(&\MP^-)-I(\MP)|\\
&=\int_{\Delta_G}\int_{\Delta_G}(H(p\circledast q)-H(p))d\MP(p)d\MP(q).
\end{align*}
\end{lemma}
\begin{proof}
It is sufficient to show this for Blackwell measures of DMCs (because we can then extend the equation to $\mathcal{MP}_b(G)$ by continuity).

Let $W$ be a DMC with input alphabet $G$. We have $I(W^-)\leq I(W)$ so $|I(\MP_W^-)-I(\MP_W)|=I(W)-I(W^-)$. From \cite[Proposition 8]{RajCont}, we have
\begin{align*}
I(W)&=\log|G|-\int_{\Delta_G}H(p)d\MP_W(p)\\
&=\log|G|-\int_{\Delta_G}\int_{\Delta_G}H(p)d\MP_W(p)d\MP_W(q).
\end{align*}

Similarly,
\begin{align*}
&I(W^-)\\
&=\log|G|-\int_{\Delta_G}H(p)d\MP_{W^-}(p)\\
&\stackrel{(a)}{=}\log|G|-\int_{\Delta_G}H(p)d(\MP_{W},\MP_W)^-(p)\\
&\stackrel{(b)}{=}\log|G|-\int_{\Delta_G}H(p)dC^{-,+}_\#(\MP_{W}\times\MP_W)(p)\\
&\stackrel{(c)}{=}\resizebox{0.45\textwidth}{!}{$\displaystyle\log|G|-\int_{\Delta_G\times\Delta_G}H(C^{-,+}(p,q))d(\MP_{W}\times \MP_W)(p,q)$}\\
&\stackrel{(d)}{=}\log|G|-\int_{\Delta_G}\int_{\Delta_G}H(p\circledast q)d\MP_{W}(p)d\MP_W(q),
\end{align*}
where (a) follows from \cite[Proposition 10]{RajCont}, (b) follows from the definition of the $(-,+)$-convolution (see Page 20 of \cite{RajCont}), (c) follows from the properties of the push-forward probability measure, and (d) follows from the definition of the $C^{-,+}$ map (see Page 20 of \cite{RajCont}) and Fubini's theorem.

The lemma now follows from the fact that $|I(W^-)-I(W)|=I(W)-I(W^-)$.
\end{proof}

For every $p\in\Delta_G$ and every $u\in G$, define $p_u\in\Delta_G$ as
$$p_{u}(x)=p(x+u).$$

Let $p,q\in\Delta_G$. We have
$$p\circledast q=\sum_{u_2\in G}q(u_2)p_{u_2}.$$
Due to the strict concavity of entropy, we have $H(p\circledast q)\geq H(p)$. Moreover, we have
\begin{equation}
\begin{aligned}
H(p\circledast q)= H(p)\;&\Leftrightarrow\;p_{u_2}=p_{u_2'}, \forall u_2,u_2'\in\supp(q)\\
&\Leftrightarrow\;p_u=p,\forall u\in\Delta\supp(q)\\
&\Leftrightarrow\;p_u=p,\forall u\in\langle\Delta\supp(q)\rangle,
\end{aligned}
\label{eqEquiv}
\end{equation}
where $\Delta\supp(q)=\{u_2-u_2':\;u_2,u_2'\in\supp(q)\}$, and $\langle\Delta\supp(q)\rangle$ is the subgroup of $G$ generated by $\Delta\supp(q)$.

\begin{lemma}
\label{lemMainCond}
Let $\MP\in\mathcal{MP}(G)$. We have $I(\MP^-)=I(\MP)$ if and only if for every $p,q\in\supp(\MP)$ we have $p=p_u$ for every $u\in\langle\Delta\supp(q)\rangle$.
\end{lemma}
\begin{proof}
Define the function $F:\Delta_G\times\Delta_G\rightarrow\mathbb{R}^+$ as
$$F(p,q)=H(p\circledast q)-H(p).$$

Since $F$ is continuous and positive, the integral
$$\int_{\Delta_G}\int_{\Delta_G}F(p,q)d\MP(p)d\MP(q)$$
is equal to zero if and only if the function $F$ is equal to zero on $\supp(\MP)\times\supp(\MP)$.

The lemma now follows from Lemma \ref{lemMainEq} and Equation \eqref{eqEquiv}.
\end{proof}

\begin{lemma}
\label{lemMainLemma}
$\mathcal{POL}_G=\{\MP_{D}:D\in\mathcal{DH}_G\}$.
\end{lemma}
\begin{proof}
Let $H$ be a subgroup of $G$. It is easy to see that $\displaystyle\MP_{D_H}=\frac{1}{|G/H|}\sum_{A\in G/H}\delta_{\pi_A},$
where $\delta_{\pi_A}$ is a Dirac measure centered at $\pi_A$ (the uniform distribution on $A$). It is easy to check that $\MP_{D_H}$ satisfies the condition of Lemma \ref{lemMainCond}, hence $I(\MP_{D_H}^-)=I(\MP_{D_H})$ and so $$\{\MP_{D}:D\in\mathcal{DH}_G\}\subset \mathcal{POL}_G.$$

Now suppose that $\MP\in \mathcal{POL}_G$. For every $p\in\supp(\MP)$, define $A_p=\supp(p)$ and $H_p=\langle \Delta A_p\rangle$. Lemma \ref{lemMainCond} shows that $p=p_u$ for every $u\in H_p$. Let $x,x'\in A_p$. We have:
$$p(x')=p(x+x'-x)=p_{x'-x}(x)\stackrel{(a)}{=}p(x),$$
where (a) follows from the fact that $x'-x\in H_p$. This shows that $p$ is the uniform distribution on $A_p$. Moreover, for every $u\in H_p$, we have
$$p(x+u)=p_{-u}(x+u)=p(x+u-u)=p(x)>0.$$
This implies that $A_p=x+H_p$, which means that the support of $p$ is a coset of the subgroup $H_p$.

Now let $p,q\in\supp(\MP)$. Let $x\in A_p$ and $u\in H_q$. Lemma \ref{lemMainCond} implies that
$$p(x+u)=p_u(x)=p(x)>0,$$
hence $u=x+u-x\in H_p$. This shows that $H_q\subset H_p$. Similarly, we can show that $H_p\subset H_q$. Therefore, $H_p=H_q$ for every $p,q\in\supp(\MP)$. This means that the support of $\MP$ consists of uniform distributions over cosets of the same subgroup of $G$. Let $H$ be this subgroup.

The above discussion shows that
$$\MP=\sum_{A\in G/H}\alpha_A\delta_{\pi_A},$$
for some distribution $\{\alpha_A:\;A\in G/H\}$ over the quotient group $G/H$. Fix $A\in G/H$ and let $x\in A$. We have:
\begin{align*}
\frac{1}{|G|}&=\pi_G(x)\stackrel{(a)}{=}\int_{\Delta_G} p(x) d\MP(p)\\
&=\sum_{B\in G/H}\alpha_B \pi_B(x)=\alpha_A \frac{1}{|A|}=\frac{\alpha_A}{|H|},
\end{align*}
where (a) follows from the fact that $\MP$ is balanced. Hence $\displaystyle\alpha_A=\frac{|H|}{|G|}=\frac{1}{|G/H|}$, and so $\displaystyle\MP=\frac{1}{|G/H|}\sum_{A\in G/H}\delta_{\pi_A}$. This means that $\MP=\MP_{D_H}$, thus $\MP\in \{\MP_{D}:D\in\mathcal{DH}_G\}$, and so $\mathcal{POL}_G\subset\{\MP_{D}:D\in\mathcal{DH}_G\}$.

We conclude that $\mathcal{POL}_G=\{\MP_{D}:D\in\mathcal{DH}_G\}$.
\end{proof}

\begin{theorem}
\label{theConvergence}
Let $(W_n)_{n\geq 0}$ be the channel-valued process defined in Definition \ref{defPolProcess}. Almost surely, there exists a subgroup $H$ of $G$ such that the sequence $(\hat{W}_n)_{n\geq 0}$ converges to $\hat{D}_H$ in the noisiness/weak-$\ast$ topology.
\end{theorem}
\begin{proof}
Lemma \ref{lemConvergence} shows that almost surely, the sequence $\big(|I(W_n^-)-I(W_n)|\big)_{n\geq 0}$ converges to zero. Let $(W_n)_{n\geq 0}$ be a sample of the process for which the sequence $\big(|I(W_n^-)-I(W_n)|\big)_{n\geq 0}$ converges to zero. Lemma \ref{lemAlmostExtreme} implies that $\big(d_{G,\ast}^{(o)}(\MP_{W_n},\mathcal{POL}_G)\big)_{n\geq 0}$ converges to zero. Now since $\mathcal{POL}_G$ is finite (see Lemma \ref{lemMainLemma}), the sequence $(\MP_{W_n})_{n\geq 0}$ converges to an element in $\mathcal{POL}_G$. Lemma \ref{lemMainLemma} now implies that there exists a subgroup $H$ of $G$ such that the sequence $(\hat{W}_n)_{n\geq 0}$ converges to $\hat{D}_H$ in the noisiness/weak-$\ast$ topology.
\end{proof}

\begin{corollary}
For any channel functional $F:\DMC_{\mathcal{X},\ast}\rightarrow\mathbb{R}$ which is invariant under channel equivalence, and which is continuous in the noisiness/weak-$\ast$ topology, the process $\big(F(W_n)\big)_{n\geq 0}$ almost surely converges. More precisely, $F(W_n)$ converges to $F(D_H)$ if $\hat{W}_n$ converges to $\hat{D}_H$.
\end{corollary}

\section{Discussion}

Our proof can be used (verbatim) to show the almost sure convergence of the polarization process associated to any channel whose equivalence class is determined by the Blackwell measure. This family of channels is large and contains almost any ``sensible" channel we can think of \cite{torgersen}.

\section*{Acknowledgment}
I would like to thank Emre Telatar and Maxim Raginsky for helpful discussions.

\bibliographystyle{IEEEtran}
\bibliography{bibliofile}

\end{document}